\documentclass[10pt,conference]{IEEEtran}
\IEEEoverridecommandlockouts

\usepackage[T1]{fontenc}
\usepackage[dvipsnames]{xcolor}
\usepackage[normalem]{ulem} % for e.g. colored strikethrough
\usepackage{float}
\usepackage{inputenc}
\usepackage{amsthm}
\usepackage{thmtools}
\usepackage{mathtools}
\usepackage{physics}
\usepackage{afterpage}
\usepackage{bm} % bold math
\usepackage{bbm} % black board math
\usepackage{braket}
\usepackage{dsfont}
\usepackage[bb=stix]{mathalpha} % blackboard math
\usepackage{amsmath}
\usepackage{amssymb}
\usepackage{tikz}
\usetikzlibrary{arrows, arrows.meta, backgrounds, calc, fit, positioning, quantikz2, shapes, shapes.geometric}
\usepackage{tikz-cd}
\usepackage{ulem}
\usepackage[linesnumbered,ruled]{algorithm2e}
\SetArgSty{textnormal}

\usepackage{etoolbox}
\apptocmd{\thebibliography}{\raggedright}{}{}

        \DeclareMathOperator{\spn}{span} % span
        \DeclareMathOperator{\id}{id} % identity
        \newcommand*{\one}{\mathds{1}} % identity operator
        \newcommand*{\N}{\mathbb{N}}

        \newcommand*{\R}{\mathbb{R}}
        \newcommand*{\C}{\mathbb{C}}
        
        \newcommand\scriptin{\raisebox{0.15ex}{$\scriptscriptstyle\in$}} % "element of" symbol for usage in sub- and superscripts
        
        \newcommand*{\bigo}{\mathcal{O}} % Big-O notation
        \newcommand*{\lo}{\mathcal{L}} % linear bounded operators
        \newcommand*{\bit}{\mathbb{b}} % single bit
        \newcommand*{\qubit}{\mathbb{q}} % single-qubit system

        \newcommand*{\cop}{\text{{\normalfont\scshape cop}}} % generic combinatorial optimisation problem
        \newcommand*{\tsp}{\text{{\normalfont\scshape tsp}}} % Travelling Salesperson Problem
        \newcommand*{\controlled}{\texttt{C}} % control
        \newcommand*{\swp}{\texttt{SWAP}} % SWAP gate
        \newcommand*{\h}{\texttt{H}} % Hadamard gate
        \newcommand*{\x}{\texttt{X}} % Pauli-X gate
        \newcommand*{\rx}{\texttt{RX}} % rotational Pauli-X gate
\usepackage[colorlinks]{hyperref}
\hypersetup{
    colorlinks  = true,
    citecolor   = PineGreen,
    linkcolor   = MidnightBlue,
    urlcolor    = PineGreen
}

\declaretheorem[style=plain]{theorem}

\declaretheorem[style=plain,sibling=theorem]{corollary}
\declaretheorem[style=plain,sibling=theorem]{lemma}

\declaretheorem[style=definition,sibling=theorem]{definition}

\begin{document}

\title{Exhaustive and feasible parametrisation with applications to the travelling salesperson problem
\thanks{This work was supported by the DFG under Germany's Excellence Strategy--EXC-2123 QuantumFrontiers, the Quantum Valley Lower Saxony and the BMFTR projects ATIQ and Quics.}}

\author{\IEEEauthorblockN{Marvin Schwiering\IEEEauthorrefmark{1}, Timo Ziegler\IEEEauthorrefmark{1}, Lennart Binkowski\IEEEauthorrefmark{1}, Benjamin Sambale\IEEEauthorrefmark{2}}
\IEEEauthorblockA{\IEEEauthorrefmark{1}Institut f\"{u}r Theoretische Physik\\
Leibniz Universit\"{a}t Hannover, Hannover, Germany\\
Email: $\{$marvin.schwiering, timo.ziegler, lennart.binkowski$\}$@itp.uni-hannover.de}\\
\IEEEauthorrefmark{2}Institut f\"{u}r Algebra, Zahlentheorie und Diskrete Mathematik,\\
Leibniz Universit\"{a}t Hannover, Hannover, Germany\\
Email: sambale@math.uni-hannover.de}

\maketitle

\begin{abstract}
This paper introduces the concept of exhaustively parametrised, feasibility-respecting quantum circuits for constrained combinatorial optimisation problems.
Such circuits can reach, given the right parameter values, every feasible solution with certainty---including the optimum---with a fixed number of parameters, while avoiding infeasible solutions altogether.
This is in sharp contrast to conventional quantum alternating operator ansatz schemes, which are merely guaranteed to reach the optimum asymptotically.
We introduce an abstract pipeline for constructing exhaustively parametrised, feasibility-respecting circuits from a transitive group action on a problem's feasible set.
Our constructions rely on the simple combination of the group action with group representation and the novel notion of generating sequences: group elements in fixed order, possibly with repetitions, that generate the entire group.
That is, we trace expressivity of parametrised quantum circuits back to the most fundamental concepts of group theory.
We apply this pipeline to two concrete examples for the travelling salesperson problem, thus showing that exhaustively parametrised, feasibility-respecting circuits are not an empty definition.
Furthermore, we provide numerical proof-of-principles on instances with up to nine cities, comparing the suitability of our constructions for parameter optimisation purposes against established mixers.
\end{abstract}


\section{\label{section:Introduction}Introduction}

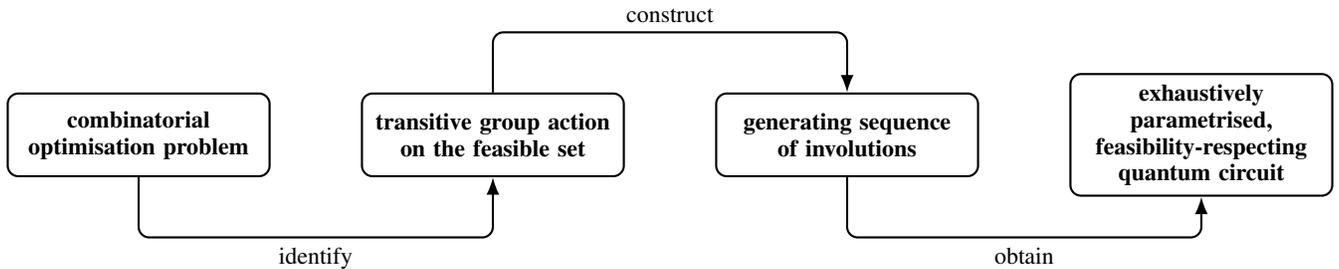
\begin{figure*}[!ht]
\begin{center}
\begin{tikzpicture}[
  >=Latex,
  node distance=10mm and 12mm,
  box/.style={
    draw,
    rounded corners,
    thick,
    align=center,
    minimum height=11mm,
    text width=32mm,
    inner sep=4pt
  },
  smallbox/.style={
    draw,
    rounded corners,
    thick,
    align=center,
    minimum height=10mm,
    text width=31mm,
    inner sep=4pt
  },
  line/.style={-Latex,thick},
  every node/.style={font=\small}
]

\node[box] (problem) {
  \textbf{combinatorial optimisation problem}
};

\node[box, right=of problem] (action) {%
  \textbf{transitive group action on the feasible set}
};

\node[box, right=of action] (sequence) {%
  \textbf{generating sequence of involutions}
};

\node[box, right= of sequence] (circuit) {%
  \textbf{exhaustively parametrised, feasibility-respecting quantum circuit}
};

\draw[line, rounded corners]
    (problem.south) -- ++(0,-0.8)
    coordinate (p1)
    -- node[midway, below] {identify}
       (p1 -| action.south)
    -- (action.south);
\draw[line, rounded corners]
    (action.north) -- ++(0,+0.8)
    coordinate (p2)
    -- node[midway, above] {construct}
       (p2 -| sequence.north)
    -- (sequence.north);
\draw[line, rounded corners]
    (sequence.south) -- ++(0,-0.8)
    coordinate (p3)
    -- node[midway, below] {obtain}
       (p3 -| circuit.south)
    -- (circuit.south);

\end{tikzpicture}
\caption{\label{figure:AbstractPipelineEasy}
    High-level overview of our abstract pipeline for constructing exhaustively parametrised, feasibility-respecting quantum circuits for combinatorial optimisation problems.
    The task of exactly and exclusively preparing all feasible solutions via paremtrised quantum circuits can be reduced to concepts of elementary group theory.
    The two key ingredients are a transitive group action on the problem's feasible set, often available from decades of operational research studying symmetries in the feasible set, and a generating sequence of involutions.
    The output of the pipeline is an exhaustively parametrised, feasibility-respecting quantum circuit with a fixed number of parameters.
}
\end{center}
\end{figure*}

The \emph{quantum approximate optimisation algorithm}~\cite{Farhi2014AQuantumApproximateOptimizationAlgorithm} and its generalisation to the \emph{quantum alternating operator ansatz} (QAOA)~\cite{Hadfield2019FromTheQuantumApproximateOptimizationAlgorithmToAQuantumAlternatingOperatorAnsatz} arguably constitute the most-studied class among (near-term) quantum algorithms for combinatorial optimisation.
In a nutshell, QAOA consists of optimising the expectation value of an objective Hamiltonian, encoding the classical objective function, over a parameterised class of quantum states.
This class is obtained by applying parameterised quantum circuits to a given initial state.
The parameter optimisation itself is often outsourced to a classical optimiser, while the quantum computer is mainly used as a sampling machine, providing fast estimations of the objective function (and its gradients).
Traditional classical algorithms for the QAOA parameter optimisation, such as COBYLA~\cite{Powell1998DirectSearchAlgorithmsForOptimizationCalculations} or SLSQP~\cite{Kraft1994Algorithm733TompFortranModulesForOptimalControlCalculations}, are usually (pseudo) gradient-based.
More recently, also gradient-free approaches based on Monte Carlo tree search have been investigated~\cite{Agirre2025AMonteCarloTreeSearchApproachToQaoaFindingANeedleInTheHaystack}. 
An increasingly popular alternative to direct parameter optimisation is to transfer well-performing parameters between similar instances~\cite{Wurtz2021FixedAngleConjecturesForTheQuantumApproximateOptimizationAlgorithmOnRegularMaxcutGraphs}.

Irrespective of the utilised parameter optimisation technique, QAOA circuits typically face reachability issues in the sense that the class of parametrised states does not (fully) cover any optimal solution state.
This is true for the original design for unconstrained and soft-constrained problems, but also often applies to refined versions which incorporate the problem's feasibility structure directly into the parametrised circuit.
While the original formulation and all versions derived from adiabatic protocols~\cite{Farhi2000QuantumComputationByAdiabaticEvolution} can at least ensure exact reachability in the (arguably practically irrelevant) limit of infinitely many variational parameters~\cite{Binkowski2024ElementaryProofOfQaoaConvergence}, even this property may vanish for generic QAOA circuits for hard-constrained problems.
In any case, this renders most QAOA circuits with a fixed number of parameters unreliable candidates for ever detecting the optimum of given optimisation problem.

Even though the heuristic nature of the QAOA might not require reliable sampling of optima in practice, we believe that the difficulty of finding good solutions with the QAOA should solely concentrate in the parameter optimisation, not in the coverage of the feasible space itself.
To this end, we introduce the concept of \emph{exhaustively parametrised} quantum circuits, which are guaranteed to reach every feasible solution state---including the optima---exactly and with a finite number of parameters.
Analogous to the generalisation of the QAOA for constrained problems, such circuits should only produce (superpositions of) feasible states, restricting the search space to the feasible subspace and thus avoiding overhead through feasibility checks and penalty terms.
Overall, our construction removes a representational obstruction.
It does not, by itself, imply efficient parameter optimisation or an end-to-end quantum speedup.

While the concept of exhaustively parametrised, feasibility-respecting quantum circuits might be straightforward, constructing one for a concrete problem at hand poses a serious engineering challenge, requiring profound knowledge about the problem's feasibility structure.
Our main contribution is the formalisation of an abstract pipeline, which essentially reduces the construction of such circuits to identifying a transitive group action on the feasible set.
Intuitively, a transitive group action is able to map between every pair of arbitrary feasible elements.
Representing the group (action) as multi-qubit operators is the conceptually easier part.
Capturing the entire group with fixed number of elements requires the novel notion of a generating sequence.

We apply our pipeline to the \emph{travelling salesperson problem} (TSP), one of the most paradigmatic NP-hard combinatorial optimisation problems with constraints.
Its feasible set can be naturally identified with the symmetric group over all city indices, suggesting to consider the action of the symmetric group on its own, i.e., on the TSP's feasible set, as entry point for our pipeline.
We derive two classes of exhaustively parametrised, feasibility-respecting circuits based on the same group action, but on different generating sequences.
The first example is inspired by the classical bubble sort algorithm.
The latter's ability to correctly sort an arbitrary input list can be reinterpreted as its building blocks---adjacency transpositions---constituting a generating sequence for the symmetric group.
The bubble sort sequence for an $n$-city instance has $\bigo(n^{2})$ elements and yields a circuit with the same number of parameters.
The second design is based on a recursive construction of generating sequences, which only adds $\log(n)$ elements when increasing $n - 1 \mapsto n$.
Applying this construction $n - 1$ times in a row yields a generating sequence with only $n \log(n)$ elements.

The TSP has been frequently in the scope of the QAOA and related methods, allowing a direct comparison between our approach and established ones.
Concretely, the seminal work~\cite{Hadfield2019FromTheQuantumApproximateOptimizationAlgorithmToAQuantumAlternatingOperatorAnsatz} enhancing the QAOA for constrained problems, also features concrete feasibility-preserving quantum circuits for the TSP.
We utilise them as baseline for numerical experiments on TSP instances with up to nine cities.
Despite increased performance of our approaches on the selected test instance, we refrain from predicting any immediate practical implications for solving realistic TSP instances.
On the one hand, the unpredictably complex parameter landscapes permit any meaningful extrapolations, neither to instances of a similar size, let alone to realistic instances.
On the other, instances with thousands of variables have been classically solved to provable optimality~\cite{Applegate2009CertificationOfAnOptimalTspTourThrough85900Cities}---a regime simply not matched by near-term quantum computing and tailored quantum algorithms.
Instead, the TSP constructions and their numerical study are to be understood as proof-of-principles of our abstract pipeline, whose impact will grow the better we understand the feasibility structure of other, more challenging optimisation problems.

\section{\label{section:PreviousWork}Previous work}

The TSP is arguably one of the most prominent combinatorial optimisation problems and motivated the development of several, more generally applicable classical frameworks.
For example, in their seminal paper~\cite{Dantzig1954SolutionOfALargeScaleTravelingSalesmanProblem}, Dantzig, Fulkerson, and Johnson already utilised cutting planes~\cite{Marchand2002CuttingPlanesInIntegerAndMixedIntegerProgramming} and branch-and-bound~\cite{Land1960AnAutomaticMethodOfSolvingDiscreteProgrammingProblems} methods.
Bellman~\cite{Bellman1962DynamicProgrammingTreatmentOfTheTravellingSalesmanProblem} developed the concept of dynamic programming and studied its application to the TSP.
Also the Held-Karp algorithm~\cite{Held1962ADynamicProgrammingApproachToSequencingProblems}, which enjoys the best asymptotic worst-case runtime among all known classical exact algorithms for the TSP, is based on dynamic programming.

In the field of quantum(-assisted) optimisation, the original formulation of the QAOA by Farhi \textit{et al.}~\cite{Farhi2014AQuantumApproximateOptimizationAlgorithm} is arguably the best-studied algorithm.
It is universally applicable to combinatorial optimisation.
However, constrained problems, such as the TSP, require a reformulation as (usually quadratic) unconstrained binary optimisation problems (QUBOs)~\cite{Kochenberger2014TheUnconstrainedBinaryQuadraticProgrammingProblemASurvey}.
This reformulation technique predates the field of quantum optimisation---for example, the QUBO formulation of the TSP can be traced back to early works of Hopfield and Tank~\cite{Hopfield1985NeuralComputationOfDecisionsInOptimizationProblems}---but regained massive popularity due to the QAOA.
Most notably, Lucas~\cite{Lucas2014IsingFormulationsOfManyNpProblems} published a comprehensive list of QUBO formulations for various combinatorial optimisation problems, including Hopfield and Tank's formulation that was previously rediscovered by Warren~\cite{Warren2012AdaptingTheTravelingSalesmanProblemToAnAdiabaticQuantumComputer}.
More space-efficient encoding, resulting in higher-order binary optimisation (HOBO) formulations, have been discussed as well~\cite{Glos2020SpaceEfficientBinaryOptimizationForVariationalComputing}.
Especially the QUBO formulation of the TSP has recently come under increasing criticism due to well--documented shortcomings in QUBO-based classical methods~\cite{SmithMiles2025TheTravellingSalespersonProblemAndTheChallengesOfNearTermQuantumAdvantage}.

As a promising alternative for constrained problems, Hadfield \textit{et al.}~\cite{Hadfield2019FromTheQuantumApproximateOptimizationAlgorithmToAQuantumAlternatingOperatorAnsatz} enhanced the QAOA to directly incorporate constraints in a problem-specific mixer.
They designed concrete mixers for several problem classes, including the TSP, while Ruan \textit{et al.}~\cite{Ruan2020TheQuantumApproximateAlgorithmForSolvingTravelingSalesmanProblem} considered pre-compiling suitable QAOA mixers for TSP.
Alternative mixers, derived via similar group-theoretical considerations as in this work, have been proposed for the open-shop scheduling problem~\cite{Binkowski2025SymmetryBasedQuantumAlgorithmsForOpenShopSchedulingWithHardConstraints}, containing the TSP as special case.
The group-theoretical treatment of constraints complements the analysis in~\cite{Shaydulin2021ClassicalSymmetriesAndTheQuantumApproximateOptimizationAlgorithm}, studying the effect of symmetries in the objective function on the QAOA dynamics.

Besides variational methods, there exist other hybrid and quantum approaches to the TSP.
Most notably, the application of the closely related framework of quantum annealing has been proposed~\cite{Martonak2004QuantumAnnealingOfTheTravelingSalesmanProblem}, extensively studied~\cite{Warren2019SolvingTheTravelingSalesmanProblemOnAQuantumAnnealer}, and fine-tuned for the TSP~\cite{Chen2006OptimizedAnnealingOfTravelingSalesmanProblemFromTheNthNearestNeighborDistribution}.
Another interesting proposal encodes city coordinates directly on the Bloch sphere of a single qubit~\cite{Goswami2025SolvingTheTravellingSalesmanProblemUsingBlochSphereEncoding}.
Furthermore, Grover's famous search algorithm~\cite{Grover1996AFastQuantumMechanicalAlgorithmForDatabaseSearch} has been adapted to the TSP, both, for feasible state preparation and subsequent search~\cite{Sato2025TwoStepQuantumSearchAlgorithmForSolvingTravelingSalesmanProblems}, based on a generalised Grover oracle~\cite{Bang2012AQuantumHeuristicAlgorithmForTheTravelingSalesmanProblem}.
Bärtschi and Eidenbenz~\cite{Bartschi2020GroverMixersForQaoaShiftingComplexityFromMixerDesignToStatePreparation}, alongside their proposal of the Grover-mixer QAOA, provided an efficient preparation routine for the uniform superposition of all TSP states, rendering Grover's algorithm obsolete for this task.
Their construction can, despite the claims in~\cite{Sato2025TwoStepQuantumSearchAlgorithmForSolvingTravelingSalesmanProblems}, be extended to the more compact HOBO encoding \cite[Corollary 3]{Binkowski2025QuantumFisherYatesShuffleUnifyingMethodsForGeneratingUniformSuperpositionsOfPermutations}.
More recently, heuristic, non-uniform state preparation routines for the TSP have been also discussed~\cite{Hess2026GroverAdaptiveSearchWithProblemSpecificStatePreparation}.

\section{\label{section:Preliminaries}Preliminaries}

Throughout this article, we use the notational short-hands $\log(n) \coloneqq \lceil\log_{2}(n)\rceil$, as well as $[n] \coloneqq \{1, \ldots, n\}$ for $n \in \N$, $\bit \coloneqq \{0, 1\}$ for the bit, and $\qubit \coloneqq \C^{2}$ for the qubit.
For $\bm{x} \in \bit^{n}$, $\abs{\bm{x}}$ denotes its Hamming weight, i.e., the number of bits of $\bm{x}$ set to $1$.

\subsection{\label{subsection:TSP}Travelling salesperson problem}

Abstractly, an instance of the \emph{Travelling Salesperson Problem} (\tsp{}) on $n$ cities is defined on a complete weighted digraph $K_{n} = (V, E, w)$ (without self-loops) with $\abs{V} = n$ and edge weights $w \colon E \to \R_{> 0}$.
Feasible solutions correspond to \emph{Hamiltonian cycles}, i.e., cycles that visit every vertex exactly once; $K_{n}$ has $(n - 1)!$ such cycles.
The goal is to find a Hamiltonian cycle of minimum total edge weight.

Instead of the linear formulation used by Dantzig, Fulkerson, and Johnson~\cite{Dantzig1954SolutionOfALargeScaleTravelingSalesmanProblem}, let us consider an alternative encoding which allows for a cleaner transfer to the quantum-mechanical treatment.
It was popularised by Lucas~\cite{Lucas2014IsingFormulationsOfManyNpProblems} for the Ising/QUBO formulation of the \tsp{} and utilises $n^{2}$ bits, each indexed by the city $u \in [n]$ and the time $t \in [n]$.
Then, $x_{u, t} = 1$ means that city $u$ is visited at time $t$.
Accordingly, if $x_{u, t} = x_{v, t + 1} = 1$ (or $x_{u, n} = x_{v, 1} = 1$), the tour traverses the edge $(u, v) \in E$, giving rise to a quadratic objective function, mapping $n^{2}$-bit strings $\bm{x}$ to
\begin{align*}
    \sum_{u = 1}^{n} \sum_{v = 1}^{n} w(u, v) \left(x_{u, n} x_{v, 1} + \sum_{t = 1}^{n - 1} x_{u, t} x_{v, t + 1}\right),
\end{align*}
summing up the weights of all traversed edges.
Note that not all bit strings $\bm{x} \in \bit^{n^{2}}$ are valid tours.
We have to enforce explicitly that each city is visited exactly once and that this happens at distinct time slots:
\begin{align*}
    \sum_{t = 1}^{n} x_{u, t} = 1\ \forall\, u \in [n] \text{ and } \sum_{u = 1}^{n} x_{u, t} = 1\ \forall\, t \in [n].
\end{align*}
Interpreting these $n^{2}$-bit strings as $n \times n$ Boolean matrices, the feasible solutions exactly correspond to permutation matrices, i.e., matrices with exactly one $1$ in each row and each column.
Especially, there are $n!$ feasible assignments in this encoding so that a single Hamiltonian cycle is necessarily represented by more than just one bit string.

This approach effectively consists of two one-hot encodings for city and time slot.
It can be straightforwardly compactified by encoding one of the indices in binary instead, reducing the bit count to $n \log(n)$~\cite{Glos2020SpaceEfficientBinaryOptimizationForVariationalComputing}.
We choose here to keep the one-hot time index and to assign visited cities in $\log(n)$-subregisters.
For now, this complicates the evaluation of the objective function, but will later result in fewer operations to traverse the feasible solutions.
Instead of checking two bits $x_{u, t}$ and $x_{v, t + 1}$, we have to check which values the two $\log(n)$-bit substrings $\bm{x}_{t}$ and $\bm{x}_{t}$ hold.
Similarly, for every $u \in [n]$, the first constraint has to evaluate all subregisters $\bm{x}_{t}$, $t \in [n]$, and check that exactly one of them holds the value of $u$ in binary.
Equivalently, feasible states are those where the cities encoded in the time-subregisters form a valid permutation of~$[n]$.
The second constraint becomes redundant in this encoding, as every time $t$, by construction, carries a visited city in its associated subregister $\bm{x}_{t}$.
We already remark that the complicated objective function will not have any impact on our circuit constructions.

Both variants introduce a $n$-fold degeneracy for encoding Hamiltonian cycles as bit strings, stemming from the fact that a time coordinate introduces an artificial ``starting city'', i.e., the city visited at time $t = 1$.
Choosing any of the $n$ cities along a Hamiltonian cycle as starting city does not alter the cycle itself, but corresponds to different bit strings.
We can remove this degeneracy by removing one city---say the $n$-th city---and one time slot from the variable space, and by implicitly connecting the cities visited at times $t \in [n - 1]$ with this city.
In the twofold one-hot encoding, the resulting objective function $\tilde{c}$ maps $(n - 1)^{2}$-bit strings $\bm{x}$ to
\begin{align*}
    \sum_{u = 1}^{n - 1}\hspace*{-1pt} w(n, u) x_{u, 1}\hspace*{-1pt} +\hspace*{-1pt} w(u, n) x_{u, n - 1}\hspace*{-2pt} +\hspace*{-3pt} \sum_{v = 1}^{n - 1} \sum_{t = 1}^{n - 2}\hspace*{-1pt} w(u, v) x_{u, t} x_{v, t + 1}.
\end{align*}
The two constraints stay the same, except that they only run over $[n - 1]$ in both indices.
These changes to objective function and constraints are entirely analogous in the compactified encoding.

\subsection{\label{subsection:ElementaryGroupTheoryAndTheSymmetricGroup}Elementary group theory and the symmetric group}

The encoded $n$-city \tsp{} can be understood as an optimisation over the symmetric group $S_{n}$ of all permutations of $n$ elements, or $n - 1$ elements if one starting city is fixed.
That is, each feasible solution $\bm{x} \in \bit^{n^{2}}$---resp.\ $\bit^{n \log(n)}$---can be uniquely identified with an elements $\sigma \in S_{n}$.
The most natural such identification is given via the representation of $S_{n}$ on $\mathbb{F}_{2}^{n}$ as permutation matrices which, when flattened to an array, exactly correspond to $n^{2}$-bit strings, as discussed before.
For the $n \log(n)$-encoding, a similarly natural connection can be established:
identify a feasible $n \log(n)$-bit string $\bm{x}$ with $\sigma \in S_{n}$ with $\text{bin}(\sigma(i)) = \bm{x}_{i}$ for all $i \in [n]$.
Special elements of $S_{n}$ which will become important later are the transpositions $\tau_{i, j}^{\hphantom{2}} \in S_{n}$ with $i, j \in [n]$ and $i \neq j$, exchanging the $i$-th and $j$-th element.
They have the useful property of being involutions, i.e., it holds that $\tau_{i, j}^{2} = \id$.

Another important, but more abstract concept is that of a group action.
The elements of a group $G$ can act on a set $X$, mapping elements from $X$ again to $X$.
One imposes that this action respects the group structure of $G$ in the sense that for all $x \in X$, it holds that $e_{G} \cdot x = x$, where $e_{G}$ is $G$'s neutral element, and for all $g, h \in G$, it is $g \cdot (h \cdot x) = (gh) \cdot x$.
A group can also act on itself via (left or right) multiplication, i.e., simply via $g \cdot h = g h$ for $g, h \in G$.
This self-action has the important property of being (\emph{sharp}) \emph{transitive}: for every $h_{1}, h_{2} \in G$, there exists (exactly) one $g \in G$ so that $g \cdot h_{1} = h_{2}$.

The last important concept is that of the permutation representation of a group which is not necessarily $S_{n}$.
Given a group~$G$ acting on a finite set $X$, one can consider the formal $\C$-span of $X$, $\langle X\rangle_{\C}$, consisting of formal $\C$-linear combinations of elements in $X$.
Then, $\langle X\rangle_{\C}$ is a $\abs{X}$-dimensional complex vector space with a fixed reference basis provided by the elements of $X$.
The action of $G$ on $X$ can be naturally extended to an action on $\langle X\rangle_{\C}$ simply via linear extension, i.e.,
\begin{align*}
    g \cdot \sum_{x \scriptin X} \alpha_{x} x \coloneqq \sum_{x \scriptin X} \alpha_{x} g \cdot x.
\end{align*}
By construction, this action respects the space's linear structure and therefore represents the elements of $G$ as linear operators on $\langle X\rangle_{\C}$, which are, expressed in the natural reference basis, simple permutation matrices.

\subsection{\label{subsection:QuantumCombinatorialOptimisation}Quantum combinatorial optimisation}

The standard way of mapping combinatorial optimisation problems with objective function $c : \bit^{m} \rightarrow \R$ to a (qubit-based) quantum computer is to simply replace every bit with a qubit and to introduce the diagonal objective Hamiltonian $C \in \lo(\qubit^{\otimes m})$ via $C \ket{\bm{x}} = c(\bm{x}) \ket{\bm{x}}$.
Then, by construction, ground states of $C$ exactly correspond to (superpositions of) minimisers of $c$.
However, if the problem is constrained, such as the \tsp{}, the ground states of $C$ need not correspond to feasible solutions.

The arguably most prominent approaches for tackling combinatorial ground state search are Farhi \textit{et al.}'s Quantum Approximate Optimisation Algorithm~\cite{Farhi2014AQuantumApproximateOptimizationAlgorithm} and Hadfield \textit{et al.}'s generalisation to the Quantum Alternating Operator Ansatz~\cite{Hadfield2019FromTheQuantumApproximateOptimizationAlgorithmToAQuantumAlternatingOperatorAnsatz} (QAOA).
In a nutshell, an instantiation of QAOA prescribes an initial state $\ket{\iota}$, a depth $p \in \N$, and two parametrised quantum circuits: the \emph{phase separator} $U_{\text{P}}(\gamma)$ and the \emph{mixer} $U_{\text{M}}(\beta)$.
The (approximate) ground state search is then restricted to the parametrised states of the form
\begin{align*}
    \ket{\bm{\beta}, \bm{\gamma}} \coloneqq U_{\text{M}}(\beta_{p}) U_{\text{P}}(\gamma_{p}) \cdots U_{\text{M}}(\beta_{1}) U_{\text{P}}(\gamma_{1}) \ket{\iota}
\end{align*}
and can thus be translated to a minimisation task over the $2 p$ parameters $\bm{\beta} = (\beta_{1}, \ldots, \beta_{p})$ and $\bm{\gamma} = (\gamma_{1}, \ldots, \gamma_{p})$.

Usually, the phase separator is of the form $U_{\text{P}}(\gamma) = e^{-i \gamma C}$ and is therefore itself diagonal in the computational basis $\{\ket{\bm{x}} \in \qubit^{\otimes m} \colon \bm{x} \in \bit^{m}\}$.
For unconstrained or soft-constrained problems\footnote{A soft constraint enters as penalty term within the objective function.
Assignments violating a soft constraint are still considered feasible, but---depending on the penalty's magnitude---might have worse objective values than strictly feasible solutions in the narrower sense.}, a common choice for the initial state and the mixer is the combination of uniform superposition of all computational basis states $\ket{\iota} = \ket{\bm{+}}$ and $U_{\text{M}}(\beta) = e^{-i \beta B}$, where $B$ is the $m$-fold sum over all single-qubit Pauli-\texttt{X} operators.
$\ket{\bm{+}}$ is the unique ground state of $- B$.
For (hard-)constrained problems with feasible set $F \subset \bit^{m}$, Hadfield \textit{et al.} require the mixer to leave the \emph{feasible subspace} $\mathcal{F} \coloneqq \spn\{\ket{\bm{x}} \colon \bm{x} \in F\}$ invariant and to provide transitions between all pairs of feasible states.
The latter condition is formalised as
\begin{align}
    \forall\, \bm{x}, \bm{y} \in F\ \exists\, r \in \N\ \exists\, \beta \in \R \colon \expval{\bm{x} | U_{\text{M}}^{r}(\beta) | \bm{y}} \neq 0.
\end{align}
Note that since the usual phase separator is diagonal in the computational basis, it trivially leaves the feasible subspace invariant.
Therefore, choosing a feasible initial state $\ket{\iota} \in \mathcal{F}$ effectively restricts all parametrised states $\ket{\bm{\beta}, \bm{\gamma}}$ to the feasible subspace.

Note that, in the language of \autoref{subsection:ElementaryGroupTheoryAndTheSymmetricGroup}, $\qubit^{\otimes m}$ is nothing but the formal $\C$-span of $\bit^{m}$ with its computational basis as natural reference basis.
Analogously, the feasible subspace $\mathcal{F}$ may be viewed as the formal $\C$-span of the feasible subset $F$ with the feasible computational basis states as reference basis elements.

We conclude the preliminaries with (a simplified version of) Hadfield \textit{et al.}'s mixer construction for the \tsp{}.
They consider the $n^{2}$-encoding of the \tsp{}, but their construction can be readily adapted to the more compact $n \log(n)$-encoding.
First, for an unordered pair of times $\{s, t\}$ and an unordered pair of cities $\{u, v\}$, define the swap partial Hamiltonian
\begin{align*}
    H_{\text{PS}, \{s, t\}, \{u, v\}} \coloneqq S_{u, s}^{+} S_{v, t}^{+} S_{u, t}^{-} S_{v, s}^{-} + S_{u, s}^{-} S_{v, t}^{-} S_{u, t}^{+} S_{v, s}^{+}
\end{align*}
with $S^{+} = \ketbra{1}{0}$ and $S^{-} = \ketbra{0}{1}$.
This Hamiltonian clearly swaps cities $u$ and $v$ between time slots $s$ and $t$.
Hence, $H_{\text{PS}, t, \{u, v\}} \coloneqq H_{\text{PS}, \{t, t + 1\}, \{u, v\}}$ swaps cities $u$ and $v$ between subsequent time slots $t$ and $t + 1$.
By summing over all unordered city pairs $\{u, v\}$, we obtain $H_{\text{PS}, t}$, swapping whatever cities are visited on subsequent time slots $t$ and $t + 1$.
The sequential ordering swap mixer is now defined as
\begin{align*}
    U_{\text{seq-PS}}(\beta) \coloneqq \prod_{t = 1}^{n} e^{-i \beta H_{\text{PS}, t}}.
\end{align*}
Hadfield \textit{et al.} also consider the simultaneous version of $U_{\text{seq-PS}}$ with an exponentiated sum of swap operators instead of a product of individually exponentiated swaps.
Since $H_{\text{PS}, s}$ and $H_{\text{PS}, t}$ do not commute if $\abs{t - s} = 1$, these two notions are not equivalent.

\section{\label{section:Methods}Methods}

\begin{figure*}
\begin{center}
\begin{tikzpicture}[
  >=Latex,
  node distance=10mm and 12mm,
  box/.style={
    draw,
    rounded corners,
    thick,
    align=center,
    minimum height=11mm,
    text width=32mm,
    inner sep=4pt
  },
  smallbox/.style={
    draw,
    rounded corners,
    thick,
    align=center,
    minimum height=10mm,
    text width=31mm,
    inner sep=4pt
  },
  line/.style={-Latex,thick},
  every node/.style={font=\small}
]

% Main pipeline
\node[box] (action) {%
  \textbf{transitive group action}\\[1mm]
  $G \curvearrowright F \subset \bit^{m}$
};

\node[box, right=of action] (perm) {%
  \textbf{representation}\\[1mm]
  $h \mapsto \hat{h} \in \lo(\mathcal{F})$
};

\node[box, right=of perm] (extend) {%
  \textbf{canonical extension}\\[1mm]
  $\hat{h} \mapsto H \in \lo(\qubit^{\otimes m})$
};

\node[box, right= of extend] (prep) {%
  \textbf{state preparation}\\[1mm]
  $W \ket{\bm{0}} \propto \ket{\bm{x}},\, \bm{x} \in F$
};

\draw[line] (action) -- (perm);
\draw[line] (perm) -- (extend);

% Generating sequence branch
\node[box, below=18mm of action] (gens) {%
  \textbf{generating sequence of involutions}\\[1mm]
  $(h_{i})_{i \scriptin [d]} \in G^{d}$
};

\node[box, right=of gens] (reps) {%
  \textbf{sequence elements as multi-qubit operators}\\[1mm]
  $H_{d}, \ldots, H_{1}$
};

\node[box, right=of reps] (params) {%
  \textbf{parametrised exponentials}\\[1mm]
  $e^{-i \theta_{d} H_{d}}, \ldots, e^{-i \theta_{1} H_{1}}$
};

\draw[line] (gens) -- (reps);
\draw[line] (reps) -- (params);
\draw[line] (extend.south) -- (reps.north);

\node[box, right=of params] (circuit) {%
  \textbf{exhaustively parametrised,}\\
  \textbf{feasibility-respecting quantum circuit}\\[1mm]
  $e^{-i \theta_{d} H_{d}} \cdots e^{-i \theta_{1} H_{1}} W$
};

\draw[line] (params) -- (circuit);
\draw[line] (prep) -- (circuit);

\end{tikzpicture}
\caption{\label{figure:AbstractPipeline}
    Abstract pipeline for constructing an exhaustively parametrised, feasibility-respecting quantum circuit from a transitive group action on the feasible set $F \subset \bit^{m}$.
    Central to this construction is the existence of a generating sequence of involutions, whose elements can be represented as $m$-qubit operators.
    Exponentiating these operators, introducing real parameters in the exponent, and prepending the state preparation of some feasible computational basis state $\ket{\bm{x}}$, completes the construction.
}
\end{center}
\end{figure*}
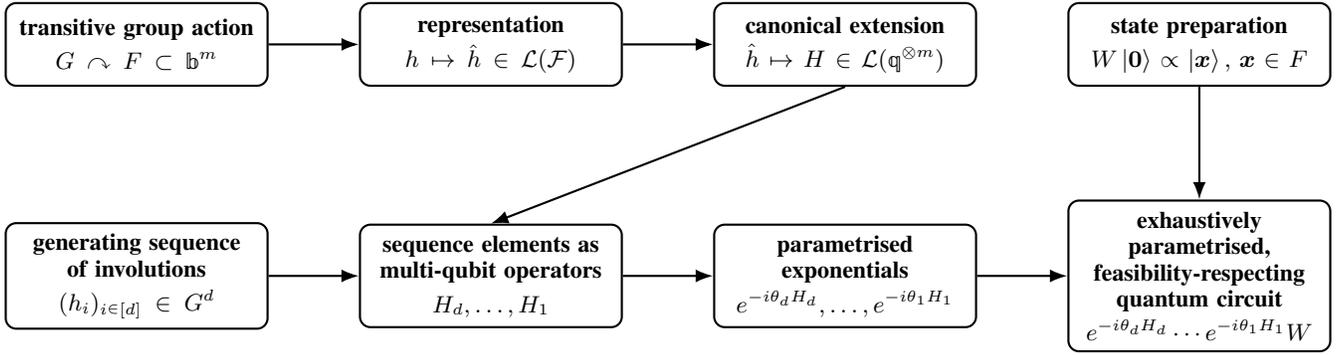

As a consequence of the adiabatic theorem, the original QAOA for unconstrained problems is guaranteed to reach the optimum in the limit of $p \to \infty$~\cite{Farhi2014AQuantumApproximateOptimizationAlgorithm}.
That is, for every $\varepsilon > 0$, there exists a $p \in \N$ and a parameter vector $(\bm{\beta}, \bm{\gamma}) \in \R^{2 p}$ so that $\abs{\expval{\bm{\beta}, \bm{\gamma} | \bm{x}^{*}}} \geq 1 - \varepsilon$, where $\bm{x}^{*} \in \bit^{n}$ is the problem's optimal solution.
This even holds true when there are multiple optimal solutions~\cite{Binkowski2024ElementaryProofOfQaoaConvergence}.
This asymptotic reachability guarantee can be extended to the QAOA for constrained problems by demanding stronger mixing properties~\cite{Binkowski2025SymmetryBasedQuantumAlgorithmsForOpenShopSchedulingWithHardConstraints}.
However, these guarantees are of little practicality;
neither can current quantum devices execute such deep circuits reliably, nor would their execution result in practical runtime advantages.
As a recourse, we formulate here the concept of exhaustively parametrised quantum circuits with non-asymptotic reachability guarantees.

\begin{definition}\label{definition:ExhaustivelyParametrisedCircuit}
    \cite[Definition 5.9]{Binkowski2026QuantumAlgorithmsForCombinatorialOptimisation}
    Let \cop{} be a combinatorial optimisation problem with feasible subset $F \subset \bit^{m}$.
    A parametrised quantum circuit $V(\bm{\theta})$, $ \bm{\theta} \in \R^{d}$, is called \emph{exhaustively parametrised} if for every $\bm{x} \in F$ there exist parameter values $\bm{\theta}_{\bm{x}}$ so that $\abs{\expval{\bm{x} | V(\bm{\theta}_{\bm{x}}) | \bm{0}}} = 1$.
\end{definition}

Especially, every optimal solution $\ket{\bm{x}^{*}}$ can be exactly reached by applying $V(\bm{\theta}_{\bm{x}^{*}})$ once to $\ket{\bm{0}}$, which in sharp contrast to the usual asymptotic guarantee.
The choice of $\ket{\bm{0}}$ as universal initial state is not restrictive, as any state preparation unitary can be absorbed into $V(\bm{\theta})$.

\autoref{definition:ExhaustivelyParametrisedCircuit} deliberately does not restrict the parametrised states to the feasible subspace.
We introduce this as a separate definition, although all examples constructed in this work conform to both requirements.

\begin{definition}\label{definition:FeasibilityRespectingCircuit}
    \cite[Definition 5.9]{Binkowski2026QuantumAlgorithmsForCombinatorialOptimisation}
    Let \cop{} be a combinatorial optimisation problem with feasible subset $F \subset \bit^{m}$.
    A parametrised quantum circuit $V(\bm{\theta})$, $ \bm{\theta} \in \R^{d}$, is called \emph{feasibility-respecting} if for every $\bm{x} \in \bit^{m} \setminus F$ and every $\bm{\theta} \in \R^{d}$ it holds that $\expval{\bm{x} | V(\bm{\theta}) | \bm{0}} = 0$.
\end{definition}

Next, we formulate a central ingredient for constructing exhaustively parametrised circuits for problems with group-based constraints, such as the \tsp{}.

\begin{definition}
    \cite[Definition 5.1]{Schwiering2024QuantumOptimizationAlgorithmsForTheTravelingSalesmanProblem}
    Let $G$ be a group and $d \in \N$.
    A sequence $(h_{i})_{i \scriptin [d]} \in G^{d}$ is a \emph{generating sequence} (\emph{for} $G$) if for every $g \in G$, there exists $\bm{b} \in \bit^{d}$ so that $g = h_{d}^{b_{d}} \cdots h_{2}^{b_{2}} h_{1}^{b_{1}}$.
\end{definition}

That is, unlike in the usual notion of group generators, elements within a generating sequence have to be applied in a fixed order and occur with a fixed multiplicity.
We give concrete examples for generating sequences for the symmetric group $S_{n}$ in the \autoref{section:Results}.

For now, consider a generic problem \cop{} with feasible set $F \subset \bit^{m}$ and a group $G$ acting on $F$.
The bit-to-qubit mapping for \cop{} naturally yields the formal $\C$-spans of $\bit^{m}$ ($\qubit^{\otimes m}$) and $F$ ($\mathcal{F}$) and hence a permutation representation of $G$ on $\mathcal{F}$.
As permutation matrices, they are especially unitary and can extended to unitaries on $\qubit^{\otimes m}$, i.e., to valid $m$-qubit quantum gates.
For us, it will not be important how the extended representations act outside of the feasible subspace, so that we may pick whatever extensions is easiest to implement.
Every such unitary extensions readily leaves the feasible subspace invariant.

As a final ingredient, we require a simple, well-known identity for involutory operators.
Let $A \in \lo(\qubit^{\otimes m})$ be a (not necessarily unitary) involution, i.e., it is $A^{2} = \one$.
Then, it holds for every $\theta \in \R$ that
\begin{align}
    e^{-i \theta A} &= \sum_{j = 0}^{\infty} \tfrac{(-i)^{2 j} \theta^{2 j} A^{2 j}}{(2 j)!} + \sum_{j = 0}^{\infty} \tfrac{(-i)^{2 j + 1} \theta^{2 j + 1} A^{2 j + 1}}{(2 j + 1)!} \nonumber \\
    &= \sum_{j = 0}^{\infty} \tfrac{(-1)^{j} \theta^{2 j}}{(2 j)!} -i \sum_{j = 0}^{\infty} \tfrac{(-1)^{j} \theta^{2 j + 1}}{(2 j + 1)!} A \nonumber \\
    &= \cos(\theta) - i \sin(\theta) A. \label{equation:ExponentialOfInvolution}
\end{align}

We are now in the position to state and prove our main abstract result.

\begin{theorem}\label{theorem:MainResult}
    Let \cop{} be a combinatorial optimisation problem with feasible subset $F \subset \bit^{m}$, let $G$ be a group acting transitively on $F$, let $(h_{i})_{i \scriptin [d]} \in G^{d}$ be a generating sequence for $G$ of involutions, and, for every $i \in [d]$, let $H_{i} \in \lo(\qubit^{\otimes m})$ be the involutory unitary so that $H_{i}\vert_{\mathcal{F}}$ is the permutation representation of $h_{i}$.
    Then, for every unitary $W \in \lo(\qubit^{\otimes m})$ with $W \ket{\bm{0}} \propto \ket{\bm{x}}$ for some $\bm{x} \in F$, the unitary
    \begin{align*}
        V(\bm{\theta}) \coloneqq e^{-i \theta_{d} H_{d}} \cdots e^{-i \theta_{2} H_{2}} e^{-i \theta_{1} H_{1}} W
    \end{align*}
    is an exhaustively parametrised, feasibility-respecting quantum circuit for \cop{}.
\end{theorem}

\begin{proof}
    Since $h_{i}$ is an involution, so is its permutation representation $\hat{h}_{i} \in \lo(\mathcal{F})$.
    Therefore, each $\hat{h}_{i}$ can be extended to a involutory unitary $H_{i} \in \lo(\qubit^{\otimes m})$---e.g., canonically by $H_{i}\vert_{\mathcal{F}^{\perp}} \coloneqq \one$.
    Now, fix $\bm{x} \in F$ and let $\bm{y} \in F$ be arbitrary.
    Since $G$ acts transitively on $F$, there exists some $g \in G$ so that $g \cdot \bm{x} = \bm{y}$.
    By the definition of a generating sequence, there exists $\bm{b} \in \bit^{d}$ so that $g = h_{d}^{b_{d}} \cdots h_{2}^{b_{2}} h_{1}^{b_{1}}$.
    As shown in \eqref{equation:ExponentialOfInvolution}, $H_{i}$ being involutory implies that $e^{-i \theta H_{i}} = \cos(\theta) \one - i \sin(\theta) H_{i}$ for all $\theta \in \R$.
    Especially, it holds that $e^{-i 0 H_{i}} = \one$ and $e^{-i \pi H_{i} / 2} = -i H_{i}$.
    Let $\bm{\theta}_{\bm{y}} \coloneqq \pi \bm{b} / 2$, then
    \begin{align*}
        V(\bm{\theta}_{\bm{y}}) \ket{\bm{0}} &= (-i H_{d})^{b_{d}} \cdots (-i H_{2})^{b_{2}} (-i H_{1})^{b_{1}} W \ket{\bm{0}} \\
        &= e^{i \xi} (-i)^{\abs{\bm{b}}} H_{d}^{b_{d}} \cdots H_{2}^{b_{2}} H_{1}^{b_{1}} \ket{\bm{x}} = e^{i \xi} (-i)^{\abs{\bm{b}}} \ket{\bm{y}}.
    \end{align*}
    That is, $\abs{\expval{\bm{y} | V(\bm{\theta}_{\bm{y}}) | \bm{0}}} = \abs{e^{i \xi} (-i)^{\abs{\bm{b}}} \expval{\bm{y} | \bm{y}}} = 1$.
    Furthermore, since $\one$ and all $H_{i}$, $i \in [d]$, leave $\mathcal{F}$ invariant, and $W \ket{\bm{0}} \in \mathcal{F}$, $V(\bm{\theta})$ is readily feasibility-respecting.
\end{proof}

\autoref{theorem:MainResult} establishes an abstract pipeline for obtaining exhaustively parametrised, feasibility-preserving quantum circuits for problems with a group acting transitively on their feasible subset and for which a generating sequence of involutions is known.
This pipeline is also schematically depicted in \autoref{figure:AbstractPipeline}.
Importantly, it only ensures exact reachability and not that the parameters realising any optimal solution can be found efficiently.

Before deriving concrete examples in \autoref{section:Results}, let us first make the general case more practical.
Suitable unitaries $W$ that prepare the initial feasible state $\ket{\bm{x}}$ can be simple Pauli-$\x$ layers, applying a bit flip to the $i$-th qubit whenever $x_{i} = 1$.
Furthermore, since $e^{-i \theta H_{i}}$ is $2 \pi$-periodic in $\theta$, and $e^{-i (\theta + \pi) H_{i}} = e^{-i \theta H_{i}} = - e^{i \theta H_{i}}$, the parameter range for $\theta$ can be effectively restricted to $[0, \pi)$.
Moreover, note that the pipeline does, unlike the QAOA, not feature the concept of a phase separator at all, but solely focusses on the feasibility structure.
This makes it simpler to transfer a construction between problems with similar feasibility features, but different objective functions, and protects against more complicated objective functions incurred by a denser encoding, as presented for the \tsp{}.
Lastly, recall the standard circuit construction for implementing an exponentiated involutory unitary $U \in \lo(\qubit^{\otimes m})$ using one ancilla qubit.
The circuit is depicted in \autoref{figure:ImplementationOfExponentiatedInvolution}.
Concretely, if $U$ is not only unitary, but also an involution, then
\begin{align*}
    e^{-i \theta U} = (\bra{\bm{0}} \h \otimes \one) \controlled(U) (\rx(2 \theta) \otimes \one) \controlled(U) (\h \ket{\bm{0}} \otimes \one)
\end{align*}
That is, the parametrised exponential can be implemented with an additional $\ket{0}$-initialised ancilla qubit and with two Hadamard gates, a parametrised rotational \texttt{X}-gate and two controlled versions of $U$.

\section{\label{section:Results}Results}

We present now two concrete constructions of exhaustively parametrised circuits for the \tsp{}.
Both will follow the abstract pipeline established by \autoref{theorem:MainResult} and the subsequent comments.
Consider an $n$-city instance in either the twofold one-hot encoding or the more compact $n \log(n)$-encoding.
For simplicity, we stick to the non-reduced case with $n!$ (instead of $(n - 1)!$) feasible solutions, but the reduced case follows the same logic.
The feasible solutions directly correspond to the elements $S_{n}$, e.g., via the identification of $n^{2}$-bit strings with $n \times n$ permutation matrices or the interpretation of $n \log(n)$-bit strings as one-line notations of permutations in binary.
Symbolically, any of these identifications establishes a bijection $\epsilon \colon S_{n} \rightarrow F$, where $F$ is the set of feasible bit strings.
The (sharp) transitive group action of $S_{n}$ on itself via left multiplication thereby extends to a (sharp) transitive action on $F$ via $\sigma \cdot \bm{x} \coloneqq \epsilon(\sigma \epsilon^{-1}(\bm{x}))$, where $\bm{x} \in F$ and $\sigma \in S_{n}$.

As the initial feasible state we may, without loss of generality, choose the bit string representing the identity permutation.
In the $n^{2}$-encoding, it is given by $x_{i, i} = 1$ and $x_{i, j} = 0$ for all $i, j \in [n]$ with $i \neq j$.
Preparing this state from $\ket{\bm{0}}$ thus requires $n$ Pauli-$\x$ gates.
In the $n \log(n)$-encoding, the bit string is given by $\bm{x}_{i} = \text{bin}(i - 1)$ for all $i \in [n]$.
This requires the application of $\approx n \log(n) / 2$ Pauli-$\x$ gates to $\ket{\bm{0}}$, since the average Hamming weight of the first $n$ integers' binary representations is $\log(n) / 2$.

\subsection{\label{subsection:BubbleSortMixer}Bubble sort sequence}

The given transitive group action on the feasible set of the \tsp{} is somewhat natural.
However, there is no canonical candidate for a generating sequence of involutory elements.
Our first example is derived from the (optimised) bubble sort algorithm~\cite{Friend1956SortingOnElectronicComputerSystems}.
Bubble sort iterates through an array of length $n$ element by element, comparing neighbouring values and swapping them if they are in the wrong oder.
It repeats this procedure, skipping comparison with the last element since it is already it is necessarily already at the correct position.
That is, the second iteration effectively operates on an array of length $n - 1$.
In each iteration, the effective length reduces by one, until it reaches $0$, meaning that the entire array is sorted.
In summary, the $i$-th iteration consists of $n - i$ adjacency comparisons and subsequent conditional swaps.
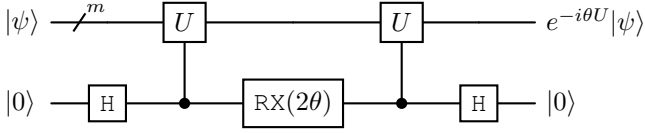
\begin{figure}[!ht]
\begin{center}
\begin{quantikz}
        \lstick{$\ket{\psi}$}	&  \qwbundle{m} & \gate[1]{U}   &               		& \gate[1]{U}   &               & \rstick{$e^{-i \theta U}\hspace*{-2pt} \ket{\psi}$} \\
        \lstick{$\ket{0}$\,}    & \gate[1]{\h}  & \ctrl{-1}     & \gate[1]{\rx(2 \theta)}       & \ctrl{-1}     & \gate[1]{\h}  & \rstick{$\ket{0}$}
\end{quantikz}
\caption{\label{figure:ImplementationOfExponentiatedInvolution}
        Implementation of $e^{-i \theta U}$, where $U$ is a Hermitian unitary.
}
\end{center}
\end{figure}
On an abstract level, and by reorganising the iterations $i \mapsto n - i$, bubble sort is able to apply all the permutations
\begin{align}\label{equation:BubbleSortSequence}
    \prod_{i = 1}^{n - 1} \left(\tau_{i}^{b_{i, i}} \circ \cdots \circ \tau_{1}^{b_{i, 1}}\right),
\end{align}
where $\tau_{i} \coloneqq \tau_{i, i + 1}$ is an adjacency transposition, and the bits $b_{i, j}$, $i \in [n - 1]$ and $j \in [i]$, determines whether two neighbouring elements $j$ and $j + 1$ are swapped in the $i$-th iteration ($b_{i, j} = 1$) or not ($b_{i, j} = 0$).
That bubble sort can sort every possible input array---or, equivalently, can invert every possible permutation of an initially sorted array---corresponds exactly to the fact that the \emph{bubble sort sequence} of adjacency transpositions in the order given by \eqref{equation:BubbleSortSequence} is a generating sequence for $S_{n}$.
We denote the bubble sort sequence for $S_{n}$ by $\Delta_{n}$ and observe that
\begin{align}\label{equation:BubbleSortSequenceLength}
    \abs{\Delta_{n}} = \tfrac{n (n - 1)}{2} \in \bigo(n^{2}).
\end{align}
Its elements are involutions so that \autoref{theorem:MainResult} is applicable.
Furthermore, by fixing a starting city, we can simply exchange $n$ with $n - 1$ in the above derivation.

Let us further discuss the representation of $\tau_{i}$, $i \in [n]$, as $m$-qubit operator, where $m \in \{n^{2}, n \log(n)\}$.
In both cases, $\tau_{i}$ also acts naturally on $\bit^{m}$ (and not only the feasible subset) by swapping the $i$-th and $(i + 1)$-th substrings.\footnote{
    To be more precise, swapping subregisters in the $n \log(n)$-encoding corresponds to a right action.
    For the reachability argument, this is not important, as also right self-actions of groups are transitive.
}
Defining $r \coloneqq m / n \in \{n, \log(n)\}$, the $i$-th substring consists of the $r$ bit positions $(i - 1) r + 1, \ldots, i r$ .
That is, $\tau_{i}$ acts as $r$ individual, disjoint transpositions of bit positions.
Each of these corresponds to a $\swp$-gate on the $m$-qubit space, yielding the representation of $\tau_{i}$ as
\begin{align}\label{equation:TranspositionRepresentation}
    \hat{\tau}_{i} = \prod_{k = 1}^{r} \swp_{(i - 1) r + k}^{i r + k}.
\end{align}

By indexing the qubits in ``column-major'' format, i.e., by distributing the $i$-th subregister over the qubits $(k - 1) n + i$, $k \in [r]$, all $\swp$-gates only act on nearest neighbours in an $m$-qubit chain.
However, implementing the parametrised exponentials $e^{-i \theta \hat{\tau}_{i}}$ via the circuit construction shown in \autoref{figure:ImplementationOfExponentiatedInvolution} requires one additional ancilla qubit with full connectivity to all other qubits, or, alternatively, $m - 1$ ancilla qubits, each of which is connected to two subsequent data qubits.
The controlled version of $\hat{\tau}_{i}$ decomposes into the $r$-fold product of controlled swaps---i.e., Fredkin gates---where the target qubits are nearest neighbours.

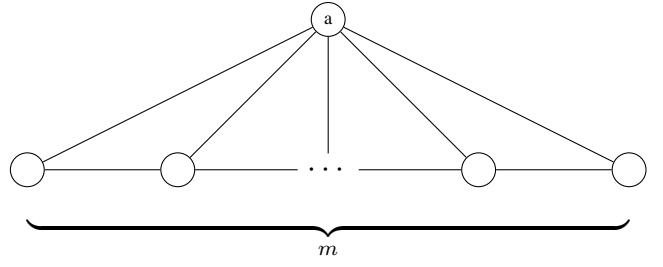
\begin{figure}[!ht]
\begin{tikzpicture}[
    every node/.style={font=\large},
    circ/.style={draw, circle, minimum size=4.5mm, inner sep=0pt}
]

% Horizontal row of 5 nodes
\node[circ] (v1) at (0,0) {};
\node[circ] (v2) at (2,0) {};
\node       (v3) at (4,0) {$\cdots$};
\node[circ] (v4) at (6,0) {};
\node[circ] (v5) at (8,0) {};

% Sixth node, vertically aligned with the middle node
\node[circ] (v6) at (4,2) {\scriptsize a};

% Horizontal edges
\draw (v1) -- (v2) -- (v3) -- (v4) -- (v5);

% Edges from the top node to all 5 lower nodes
\draw (v6) -- (v1);
\draw (v6) -- (v2);
\draw (v6) -- (v3);
\draw (v6) -- (v4);
\draw (v6) -- (v5);

% Underbrace with annotation m
\node at (4,-0.9) {$\underbrace{\hspace{8cm}}_{m}$};

\end{tikzpicture}
\caption{\label{figure:Connectivity}
    Qubit connectivity required for implementing the exhaustively parametrised quantum circuit based on the bubble sort sequence.
    The $m$ data qubits, encoding the \tsp{} solutions, can be aligned on a chain.
    The ancilla qubit, serving as control for the Fredkin gates, has to be connected to all $m$ data qubits.
    Alternatively, one can introduce $m - 1$ ancilla qubits, each connected to two contiguous data qubits.
}
\end{figure}

\begin{corollary}\label{corollary:ExhaustiveBubbleSortCircuit}
    For the $n$-city \tsp{}, there exists an exhaustively parametrised quantum circuit with $(n - 1) (n - 2) / 2$ parameters, even on star-augmented nearest-neighbour architectures (see \autoref{figure:Connectivity}).
\end{corollary}

\subsection{\label{subsection:BinaryInsertionSequence}Binary insertion sequence}

Among generating sequences, only containing adjacency transpositions, the bubble sort sequence is already of optimal length.\footnote{
Especially, the element $\mu \in S_{n}$ with $\mu(i) = n + 1 - i$ for all $i \in [n]$ has to be generated.
In the language of inversion numbers~\cite{Mannila1985MeasuresOfPresortednessAndOptimalSortingAlgorithms}: $\mu$ has inversion number $n (n - 1) / 2$, while $\id$ has inversion number $0$.
Multiplying a permutation with an adjacency transposition increases or decreases its inversion number by $1$.
Hence, reaching $\mu$ from $\id$ only with multiplications by adjacency transpositions requires at least $n (n - 1) / 2$ of them.}
However, we can find shorter generating sequences, containing also other involutions than adjacency transpositions.
Our construction relies on grouping together disjoint (non-adjacency) transpositions so that their product is again an involution.
The approach is similarly inductive on $n$ as the construction of the bubble sort sequence $\Delta_{n}$.
The latter arises from $\Delta_{n - 1}$ by embedding it into $S_{n}$ via the index increment $\tilde{\sigma}(i) = \sigma(i - 1) + 1$ and $\tilde{\sigma}(1) = 1$ for $\sigma \in S_{n - 1}$, and by subsequently appending the $n - 1$ adjacency transpositions $\tau_{i} \in S_{n}$, $i \in [n - 1]$.
These $n - 1$ additional elements ensure that $1$ can be mapped to every other position $j \in [n] \setminus \{1\}$.
However, by grouping together multiple transpositions in a single permutation and by sacrificing adjacency, one can achieve the same with less elements.

\begin{lemma}\label{lemma:AppendingLemma}
	Let $(\rho_{i})_{i \scriptin [t]}$ be a generating sequence for $S_{n - 1}$ and define, for $\ell \in [\log(n)]$, $\pi_{\ell} \in S_{n}$ so that $\pi_{\ell}(j) = j + 2^{\ell - 1}$ and $\pi_{\ell}(j + 2^{\ell - 1}) = j$ for all $j \in [\min\{n - 2^{\ell - 1}, 2^{\ell - 1}\}]$.
	Then, $(\tilde{\rho}_{1}, \ldots, \tilde{\rho}_{t}, \pi_{1}, \ldots, \pi_{\log(n)})$ is a generating sequence for $S_{n}$.
\end{lemma}

\begin{proof}
	Let $\beta \in S_{n}$ be arbitrary and let $\bm{p} \in \bit^{\log(n)}$ be the binary representation of $\beta(1) - 1$.
	Furthermore, define the permutation $\bm{\pi}^{\bm{p}} \coloneqq \pi_{\log(n)}^{p_{\log(n)}} \cdots \pi_{1}^{p_{1}}$.
	We claim that
	\begin{align}\label{equation:Claim}
		\Big(\big(\bm{\pi}^{\bm{p}}\big)^{-1} \beta\Big)(1) = 1.
	\end{align}
	Provided \eqref{equation:Claim} holds, $(\rho_{i})_{i \scriptin [t]}$ being a generating sequence for $S_{n - 1}$ implies that there exists $\bm{q} \in \bit^{t}$ so that
	\begin{align*}
		\tilde{\rho}_{t}^{q_{t}} \cdots \tilde{\rho}_{1}^{q_{1}} = \big(\bm{\pi}^{\bm{p}}\big)^{-1} \beta \implies \beta = \pi_{\log(n)}^{p_{\log(n)}} \cdots \pi_{1}^{p_{1}} \tilde{\rho}_{t}^{q_{t}} \cdots \rho_{1}^{q_{1}}.
	\end{align*}
	Since $\beta \in S_{n}$ was chosen arbitrarily, this already proves the assertion.
\end{proof}

In summary, \autoref{lemma:AppendingLemma} establishes that generating sequences for $S_{n - 1}$ can be extended to generating sequence for $S_{n}$ with only $\log(n)$ additional elements.
In addition, these elements are involutions.
Starting from the empty sequence $()$, which is (by convention) generating for $S_{1} = \{\id\}$, we can now simply apply \autoref{lemma:AppendingLemma} $n - 1$ times in order to obtain a sequence $\Xi_{n}$---which we name \emph{binary insertion sequence}---for $S_{n}$ of length
\begin{align}\label{equation:TBDSequenceLength}
	\abs{\Xi_{n}} = \sum_{i = 2}^{n} \log(i) \in \bigo(n \log(n)),
\end{align}
only consisting of involutions.
Therefore, \autoref{theorem:MainResult} is again applicable.
Notably, since the $n$-city \tsp{} has $\bigo(n!)$ feasible solutions, it follows from Stirling's formula that the length of $\Xi_{n}$ is asymptotically optimal.

We again discuss the representation of the sequence elements $\pi_{\ell}$.
They are comprised by $\min\{n - 2^{\ell - 1}, 2^{\ell - 1}\}$ individual, disjoint transpositions $\tau_{j, j + 2^{\ell - 1}}$.
Analogous to the adjacency transpositions, each of those will be represented as a subregister swap, i.e., by an $r$-fold product of $\swp$-gates.
Since the transpositions $\tau_{j, j + 2^{\ell - 1}}$ act on mutually distinct elements, their fully decomposed representation consists of $r \min\{n - 2^{\ell - 1}, 2^{\ell - 1}\}$ disjoint $\swp$-gates:
\begin{align}\label{equation:InvolutionRepresentation}
    \hat{\pi}_{\ell} = \prod_{j = 1}^{\min\{n - 2^{\ell - 1}, 2^{\ell - 1}\}} \prod_{k = 1}^{r} \swp_{(j - 1) r + k}^{(j + 2^{\ell - 1} - 1) r + k}
\end{align}

Since the elements $\pi_{\ell}$ in $\Xi_{n}$ do not correspond to adjacency transpositions anymore, it requires better qubit connectivity than the one depicted in \autoref{figure:Connectivity}, which suffices to implement the bubble sort version.

\begin{corollary}\label{corollary:ExhaustiveTBDCircuit}
	For the $n$-city \tsp{}, there exists an exhaustively parametrised quantum circuit with $\log(2) + \ldots + \log(n)$ parameters.
\end{corollary}

\subsection{\label{subsection:NumericalComparison}Numerical comparison}

After having established exact reachability guarantees for the circuits based on the bubble sort and the binary insertion sequence, we numerically study how well their parameters can be optimised for a concrete example.
We consider a 9-city \tsp{} instance whose weight matrix can be found in the \href{https://github.com/Timo59/qce26_numerics}{companion code base}.
Using the $n \log(n)$-encoding, we require $(9 - 1) \log(9 - 1) = 24$ qubits for representing the instance, effectively reducing the number of cities to $n = 8$.
We benchmark our two constructions against the QAOA with Hadfield \textit{el al.}'s \tsp{} mixers and the standard phase separator.
For our proof-of-principle, we assume all qubit and circuit components to be noiseless, all qubits to be fully connected, and work with exact expectation values rather than sample-based approximations.
We conduct the simulation using Ziegler's low-level quantum simulator Orkan~\cite{Ziegler2026OrkanCacheFriendlySimulationOfQuantumOperationsOnHermitianOperators}.

Recalling our two constructions for the \tsp{}, one can easily verify that both circuits, despite differing numbers of parameters, have the exact same circuit length.
The binary insertion sequence simply groups together multiple transpositions in one element.
One QAOA layer, in turn, contains only $n - 1$ exponentiated transpositions and the phase separator.
As a benevolent assumption, we assume the phase separator to have the same circuit length as one exponentiated transposition.
Hence $(n - 1) / 2$ QAOA layers have roughly the same circuit length as the two exhaustively parametrised circuits with their $n (n - 1) / 2$ exponentiated transpositions.
Prescribing the circuit length in the above fashion, the QAOA thus has $n - 1$ variational parameters, the binary insertion circuit has $\log(2) + \ldots + \log(n)$ parameters, and the bubble sort circuit has $n (n - 1) / 2$ parameters.
The parameter optimisation is carried out by the COBYLA optimiser~\cite{Powell1998DirectSearchAlgorithmsForOptimizationCalculations}, a pseudo gradient-based optimiser.
For all three methods, the objective Hamiltonian's expectation value serves as the objective to be optimised and recorded.
If the numerical gradient of several subsequent iterations falls short of a small, predefined threshold, we terminate the procedure.

\autoref{figure:Numerics} depicts the curves of obtained expectation values, normalised by the optimal value to yield approximation ratios.
Besides our two exhaustively parametrised circuits, we tested QAOA with two different initial states: a single feasible computational basis state and the uniform superposition of all feasible states.
For the chosen instance, both QAOA variants fail to deliver notably improved approximation ratios, while the two exhaustive constructions perform significantly better.
The bubble sort-based circuit plateaus at an approximation ratio around $0.83$, while the circuit based on binary insertion even reaches a value of roughly $0.91$.
That is, despite the exact reachability guarantee of both circuits, the optimal parameters are (unsurprisingly) not found by the (local) optimiser.
We attest the better performance of the binary insertion-based variant over the bubble sort version to the former's lower-dimensional parameter space, allowing the optimiser to navigate more efficiently.

\begin{figure}
\begin{center}
\includegraphics[width=1\linewidth]{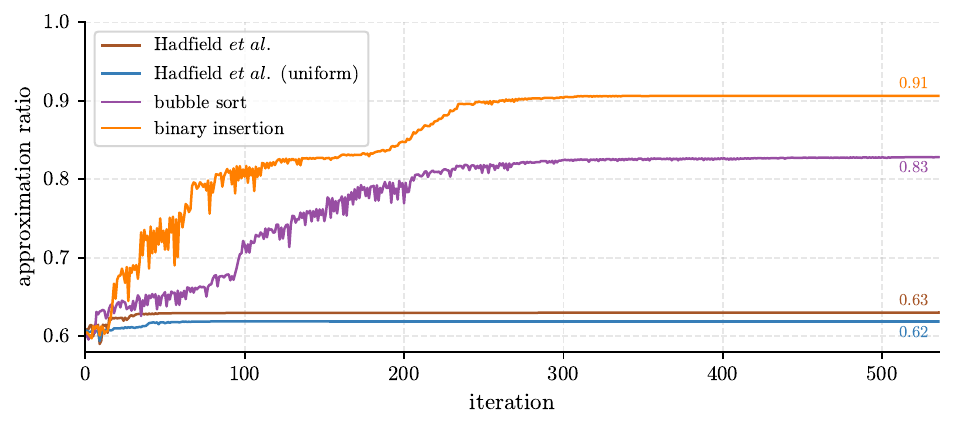}
\caption{\label{figure:Numerics}
    Approximation ratios obtained after each COBYLA iteration, using different parametrised quantum circuits on a $9$-city \tsp{} instance. 
    The QAOA with Hadfield \textit{et al.}'s mixers was not able to significantly improve initial state's approximation ratio, neither if a single computational basis state, nor the uniform superposition of all feasible states, was used.
    In contrast, the exhaustively parametrised circuit based on bubble sort was able to push the initial ratio of $\approx 0.60$ to $0.83$, while the construction based on the binary insertion sequence even reached a final approximation ratio of $0.91$.
    Not depicted here: the bubble sort-inspired approach did not terminate before 2300 steps, but was not able to significantly improve its approximation ratio during this long runtime.
    In contrast, the binary insertion variant and the QAOA initialised in a single computational basis state terminated after roughly $600$ steps.
    The QAOA with uniform initial state already terminated after roughly $150$ iterations, but failed to deliver a high-quality approximation ratio.
}
\end{center}
\end{figure}

\section{\label{section:Conclusion}Conclusion}

In this work, we have introduced the concept of exhaustively parametrised, feasibility-respecting quantum circuits, which, by design, are able to exactly reach all feasible solutions of a given optimisation problem with finitely many parameters, including optimal states, while avoiding infeasible states entirely.
These non-asymptotic reachability guarantees are unmatched by generic QAOA circuits.
We have further established an abstract pipeline for promoting a transitive group action on a problem's feasible set to such an exhaustively parametrised, feasibility-respecting circuit.
The construction is centred around the simple, yet powerful group-theoretical concept of generating sequences.
We showed that generating sequences yield natural candidates for highly expressive quantum circuits.
Additionally demanding the sequence elements to be involutions allows us to continuously parametrised these circuits, reproducing exact application or omission of individual sequence elements.

For the $n$-city Travelling Salesperson Problem (TSP), whose feasible set can be naturally identified with the symmetric group $S_{n}$, we have introduced two concrete classes of exhaustively parametrised, feasibility-respecting circuits.
Both are based on the (sharply) transitive group action of $S_{n}$ on itself, and only differ in the chosen generating sequence.
The first example consists of $\bigo(n^{2})$ adjacency transpositions and is inspired by the bubble sort algorithm.
The second example is more involved, but only contains $\bigo(n \log(n))$ elements, drastically reducing the number of parameters necessary for exact reachability.

As a proof-of-principle, we have tested both constructions against the QAOA with established TSP-mixers on a $9$-city instance.
For our concrete example, we observed improved performance of both exhaustively parameterised circuits over a $7$-layer QAOA circuit.
Moreover, the circuit with $\bigo(n \log(n))$ parameters performed notably better than the bubble sort-based circuit with its $\bigo(n^{2})$ parameters, which we attribute to the classical solver's ability to navigate better in lower-dimensional parameter spaces.
However, none of the two methods was able---despite the theoretical reachability guarantee---to recover the optimum exactly.
This underscores once more that reachability does not imply efficient trainability.
Ultimately, our proof-of-principle does not substitute for extensive benchmarking, and the investigating the implications of exact reachability to the selection and adaptation of the classical solver remains an interesting open research question.

We highlight that the pipeline utilised to find concrete circuits for the TSP is general enough to be applicable to other problems as well.
Abstractly, any group acting transitively on a problem's feasible set and has a generating sequence of involutions can be promoted to an exhaustively parametrised, feasibility-respecting quantum circuit.
A promising research direction will be to identify such group actions for other, more challenging combinatorial optimisation problems and to apply the pipeline to those problems.
Here, versions of facility location problems and vehicle routing seem to have richer symmetries in their feasible sets than packing problems.
Possible applications of our framework also arise when a group action only covers part of the feasible set.
In such cases, the pipeline might be enhanceable by additional, conditional operations to complete coverage of the entire feasible set.

\appendix

\begin{proof}[Proof of \eqref{equation:Claim}]
	For $\ell \in [\log(n)]$, define $\bm{\pi}^{\bm{p}}_{r} \coloneqq \pi_{\ell}^{p_{\ell}} \cdots \pi_{1}^{p_{1}}$ so that $\bm{\pi}^{\bm{p}}_{\log(n)} = \bm{\pi}^{\bm{p}}$.
	We prove that
	\begin{align*}
		\bm{\pi}^{\bm{p}}_{\ell}(1) = \sum_{i = 1}^{\ell} p_{i} 2^{i - 1} + 1
	\end{align*}
	for all $\ell \in [\log(n)]$ via (finite) induction on $\ell$.
	For $\ell = \log(n)$, we then obtain that
	\begin{align*}
		\bm{\pi}^{\bm{p}}(1) = \sum_{i = 1}^{\log(n)} p_{i} 2^{i - 1} + 1 = \beta(1) - 1 + 1 = \beta(1),
	\end{align*}
	since $\bm{p}$ is the binary representation of $\beta(1) - 1$.
	This identity already establishes \eqref{equation:Claim} via multiplication with $(\bm{\pi}^{\bm{p}})^{-1}$ from the left.

	\textit{Base case}.
	Let $\ell = 1$.
	Then it is $\bm{\pi}^{\bm{p}}_{\ell}(1) = \pi_{1}^{p_{1}}(1)$.
	Since $\min\{n - 2^{\ell - 1}, 2^{\ell - 1}\} = \min\{n - 1, 1\} = 1$ for all $n \geq 2$, it is $\pi_{1}(1) = 2$, $\pi_{1}(2) = 1$, and $\pi_{1}(j) = j$ for all $j \geq 3$.
	Hence,
	\begin{align*}
		\pi_{1}^{p_{1}}(1) = 1 + p_{1} 1 = p_{1} 2^{0} + 1 = \sum_{i = 1}^{1} p_{i} 2^{i - 1} + 1.
	\end{align*}

	\textit{Induction step}.
	Let $\ell > 1$ and assume that
	\begin{align*}
		\bm{\pi}^{\bm{p}}_{\ell - 1}(1) = \sum_{i = 1}^{\ell - 1} p_{i} 2^{i - 1} + 1
	\end{align*}
	holds.
	If $p_{\ell} = 0$, then the claim is trivially true also for $\ell$.
	Therefore, consider only the case of $p_{\ell} = 1$.
	It is
	\begin{align*}
		\bm{\pi}^{\bm{p}}_{\ell}(1) = \big(\pi_{\ell}^{p_{\ell}} \bm{\pi}^{\bm{p}}_{\ell - 1}\big)(1) = \pi_{\ell}\left(\sum_{i = 1}^{\ell - 1} p_{i} 2^{i - 1} + 1\right) \eqqcolon j_{0},
	\end{align*}
	where we have already used the induction hypothesis.
	Observe that $j_{0} \leq 2^{\ell - 1}$ (with equality only if $p_{i} = 1$ for all $i \in [\ell - 1]$).
	Furthermore, it is
	\begin{align*}
		j_{0} + 2^{\ell - 1} = j_{0} + p_{\ell} 2^{\ell - 1} = \sum_{i = 1}^{\ell} p_{i} 2^{i - 1} + 1 = \beta(1) \leq n.
	\end{align*}
	In summary, $j_{0} \leq \min\{n - 2^{\ell - 1}, 2^{\ell - 1}\}$, so that $\pi_{\ell}$ acts non-trivially on $j_{0}$ by incrementing it by $2^{\ell - 1} = p_{\ell} 2^{\ell - 1}$, yielding the correct result.
\end{proof}

\section*{Acknowledgment}

We thank Tim Heine, Tobias J. Osborne, Debora Ramacciotti, and Sören Wilkening for helpful discussions.

\noindent\textbf{Data and code availability statement.}
The depicted data can be found at \url{https://github.com/Timo59/qce26_numerics}.

\IEEEtriggeratref{19}
\bibliographystyle{IEEEtran}
\bibliography{IEEEabrv,bibliography}

\begin{thebibliography}{10}
\providecommand{\url}[1]{#1}
\csname url@samestyle\endcsname
\providecommand{\newblock}{\relax}
\providecommand{\bibinfo}[2]{#2}
\providecommand{\BIBentrySTDinterwordspacing}{\spaceskip=0pt\relax}
\providecommand{\BIBentryALTinterwordstretchfactor}{4}
\providecommand{\BIBentryALTinterwordspacing}{\spaceskip=\fontdimen2\font plus
\BIBentryALTinterwordstretchfactor\fontdimen3\font minus
  \fontdimen4\font\relax}
\providecommand{\BIBforeignlanguage}[2]{{%
\expandafter\ifx\csname l@#1\endcsname\relax
\typeout{** WARNING: IEEEtran.bst: No hyphenation pattern has been}%
\typeout{** loaded for the language `#1'. Using the pattern for}%
\typeout{** the default language instead.}%
\else
\language=\csname l@#1\endcsname
\fi
#2}}
\providecommand{\BIBdecl}{\relax}
\BIBdecl

\bibitem{Farhi2014AQuantumApproximateOptimizationAlgorithm}
E.~Farhi, J.~Goldstone, and S.~Gutmann, ``{A Quantum Approximate Optimization
  Algorithm},'' 2014, [arXiv preprint
  \href{https://arxiv.org/abs/1411.4028}{arXiv:1411.4028}].

\bibitem{Hadfield2019FromTheQuantumApproximateOptimizationAlgorithmToAQuantumAlternatingOperatorAnsatz}
\BIBentryALTinterwordspacing
S.~Hadfield, Z.~Wang, B.~O’Gorman, E.~G. Rieffel, D.~Venturelli
  \emph{et~al.}, ``{From the Quantum Approximate Optimization Algorithm to a
  Quantum Alternating Operator Ansatz},'' \emph{Algorithms}, vol.~12, no.~2,
  p.~34, 2019. [Online]. Available: \url{https://doi.org/10.3390/a12020034}
\BIBentrySTDinterwordspacing

\bibitem{Powell1998DirectSearchAlgorithmsForOptimizationCalculations}
\BIBentryALTinterwordspacing
M.~J.~D. Powell, ``{Direct search algorithms for optimization calculations},''
  \emph{Acta Numer.}, vol.~7, pp. 287--336, 1998. [Online]. Available:
  \url{https://doi.org/10.1017/s0962492900002841}
\BIBentrySTDinterwordspacing

\bibitem{Kraft1994Algorithm733TompFortranModulesForOptimalControlCalculations}
\BIBentryALTinterwordspacing
D.~Kraft, ``{Algorithm 733: TOMP–Fortran modules for optimal control
  calculations},'' \emph{ACM Trans. Math. Softw.}, vol.~20, no.~3, pp.
  262--281, 1994. [Online]. Available:
  \url{https://doi.org/10.1145/192115.192124}
\BIBentrySTDinterwordspacing

\bibitem{Agirre2025AMonteCarloTreeSearchApproachToQaoaFindingANeedleInTheHaystack}
\BIBentryALTinterwordspacing
A.~Agirre, E.~van Nieuwenburg, and M.~M. Wauters, ``{A Monte Carlo Tree Search
  approach to QAOA: finding a needle in the haystack},'' \emph{New J. Phys.},
  vol.~27, no.~4, p. 043014, 2025. [Online]. Available:
  \url{https://doi.org/10.1088/1367-2630/adc765}
\BIBentrySTDinterwordspacing

\bibitem{Wurtz2021FixedAngleConjecturesForTheQuantumApproximateOptimizationAlgorithmOnRegularMaxcutGraphs}
\BIBentryALTinterwordspacing
J.~Wurtz and D.~Lykov, ``{Fixed-angle conjectures for the quantum approximate
  optimization algorithm on regular MaxCut graphs},'' \emph{Phys. Rev.}, vol.
  104, no.~5, 2021. [Online]. Available:
  \url{https://doi.org/10.1103/physreva.104.052419}
\BIBentrySTDinterwordspacing

\bibitem{Farhi2000QuantumComputationByAdiabaticEvolution}
E.~Farhi, J.~Goldstone, S.~Gutmann, and M.~Sipser, ``{Quantum Computation by
  Adiabatic Evolution},'' 2000, [arXiv preprint
  \href{https://arxiv.org/abs/quant-ph/0001106}{arXiv:quant-ph/0001106}].

\bibitem{Binkowski2024ElementaryProofOfQaoaConvergence}
\BIBentryALTinterwordspacing
L.~Binkowski, G.~Ko{\ss}mann, T.~Ziegler, and R.~Schwonnek, ``{Elementary proof
  of QAOA convergence},'' \emph{New J. Phys.}, vol.~26, no.~7, p. 073001, 2024.
  [Online]. Available: \url{https://doi.org/10.1088/1367-2630/ad59bb}
\BIBentrySTDinterwordspacing

\bibitem{Applegate2009CertificationOfAnOptimalTspTourThrough85900Cities}
\BIBentryALTinterwordspacing
D.~L. Applegate, R.~E. Bixby, V.~Chv{\'a}tal, W.~Cook, D.~G. Espinoza
  \emph{et~al.}, ``{Certification of an optimal TSP tour through 85,900
  cities},'' \emph{Oper. Res. Lett.}, vol.~37, no.~1, pp. 11--15, 2009.
  [Online]. Available: \url{https://doi.org/10.1016/j.orl.2008.09.006}
\BIBentrySTDinterwordspacing

\bibitem{Dantzig1954SolutionOfALargeScaleTravelingSalesmanProblem}
\BIBentryALTinterwordspacing
G.~Dantzig, R.~Fulkerson, and S.~Johnson, ``{Solution of a Large-Scale
  Traveling-Salesman Problem},'' \emph{J. Oper. Res. Soc. Am.}, vol.~2, no.~4,
  pp. 393--410, 1954. [Online]. Available:
  \url{https://doi.org/10.1287/opre.2.4.393}
\BIBentrySTDinterwordspacing

\bibitem{Marchand2002CuttingPlanesInIntegerAndMixedIntegerProgramming}
\BIBentryALTinterwordspacing
H.~Marchand, A.~Martin, R.~Weismantel, and L.~Wolsey, ``{Cutting planes in
  integer and mixed integer programming},'' \emph{Discret. Appl. Math.}, vol.
  123, no. 1-3, pp. 397--446, 2002. [Online]. Available:
  \url{https://doi.org/10.1016/s0166-218x(01)00348-1}
\BIBentrySTDinterwordspacing

\bibitem{Land1960AnAutomaticMethodOfSolvingDiscreteProgrammingProblems}
\BIBentryALTinterwordspacing
A.~H. Land and A.~G. Doig, ``{An Automatic Method of Solving Discrete
  Programming Problems},'' \emph{Econometrica}, vol.~28, no.~3, p. 497, 1960.
  [Online]. Available: \url{https://doi.org/10.2307/1910129}
\BIBentrySTDinterwordspacing

\bibitem{Bellman1962DynamicProgrammingTreatmentOfTheTravellingSalesmanProblem}
\BIBentryALTinterwordspacing
R.~Bellman, ``{Dynamic Programming Treatment of the Travelling Salesman
  Problem},'' \emph{J. ACM}, vol.~9, no.~1, pp. 61--63, 1962. [Online].
  Available: \url{https://doi.org/10.1145/321105.321111}
\BIBentrySTDinterwordspacing

\bibitem{Held1962ADynamicProgrammingApproachToSequencingProblems}
\BIBentryALTinterwordspacing
M.~Held and R.~M. Karp, ``{A Dynamic Programming Approach to Sequencing
  Problems},'' \emph{J. Soc. Ind. Appl. Math.}, vol.~10, no.~1, pp. 196--210,
  1962. [Online]. Available: \url{https://doi.org/10.1137/0110015}
\BIBentrySTDinterwordspacing

\bibitem{Kochenberger2014TheUnconstrainedBinaryQuadraticProgrammingProblemASurvey}
\BIBentryALTinterwordspacing
G.~Kochenberger, J.~Hao, F.~Glover, M.~Lewis, Z.~L{\"u} \emph{et~al.}, ``{The
  unconstrained binary quadratic programming problem: a survey},'' \emph{J.
  Comb. Optim.}, vol.~28, no.~1, pp. 58--81, 2014. [Online]. Available:
  \url{https://doi.org/10.1007/s10878-014-9734-0}
\BIBentrySTDinterwordspacing

\bibitem{Hopfield1985NeuralComputationOfDecisionsInOptimizationProblems}
\BIBentryALTinterwordspacing
J.~J. Hopfield and D.~W. Tank, ``{“Neural” computation of decisions in
  optimization problems},'' \emph{Biol. Cybern.}, vol.~52, no.~3, pp. 141--152,
  1985. [Online]. Available: \url{https://doi.org/10.1007/bf00339943}
\BIBentrySTDinterwordspacing

\bibitem{Lucas2014IsingFormulationsOfManyNpProblems}
\BIBentryALTinterwordspacing
A.~Lucas, ``{Ising formulations of many NP problems},'' \emph{Front. Phys.},
  vol.~2, p.~5, 2014. [Online]. Available:
  \url{https://doi.org/10.3389/fphy.2014.00005}
\BIBentrySTDinterwordspacing

\bibitem{Warren2012AdaptingTheTravelingSalesmanProblemToAnAdiabaticQuantumComputer}
\BIBentryALTinterwordspacing
R.~H. Warren, ``{Adapting the traveling salesman problem to an adiabatic
  quantum computer},'' \emph{Quantum Inf. Process.}, vol.~12, no.~4, pp.
  1781--1785, 2012. [Online]. Available:
  \url{https://doi.org/10.1007/s11128-012-0490-8}
\BIBentrySTDinterwordspacing

\bibitem{Glos2020SpaceEfficientBinaryOptimizationForVariationalComputing}
A.~Glos, A.~Krawiec, and Z.~Zimbor{\'a}s, ``{Space-efficient binary
  optimization for variational computing},'' 2020, [arXiv preprint
  \href{https://arxiv.org/abs/2009.07309}{arXiv:2009.07309}].

\bibitem{SmithMiles2025TheTravellingSalespersonProblemAndTheChallengesOfNearTermQuantumAdvantage}
\BIBentryALTinterwordspacing
K.~A. Smith-Miles, H.~H. Hoos, H.~Wang, T.~B{\"a}ck, and T.~J. Osborne, ``{The
  travelling salesperson problem and the challenges of near-term quantum
  advantage},'' \emph{Quantum Sci. Technol.}, vol.~10, no.~3, p. 033001, 2025.
  [Online]. Available: \url{https://doi.org/10.1088/2058-9565/add61d}
\BIBentrySTDinterwordspacing

\bibitem{Ruan2020TheQuantumApproximateAlgorithmForSolvingTravelingSalesmanProblem}
\BIBentryALTinterwordspacing
Y.~Ruan, S.~Marsh, X.~Xue, Z.~Liu, and J.~Wang, ``{The Quantum Approximate
  Algorithm for Solving Traveling Salesman Problem},'' \emph{Comput. Mater.
  \&amp; Contin.}, vol.~63, no.~3, pp. 1237--1247, 2020. [Online]. Available:
  \url{https://doi.org/10.32604/cmc.2020.010001}
\BIBentrySTDinterwordspacing

\bibitem{Binkowski2025SymmetryBasedQuantumAlgorithmsForOpenShopSchedulingWithHardConstraints}
\BIBentryALTinterwordspacing
L.~Binkowski, G.~Ko{\ss}mann, C.~Tutschku, and R.~Schwonnek, ``{Symmetry-based
  quantum algorithms for open-shop scheduling with hard constraints},''
  \emph{Acad. Quantum}, vol.~2, no.~3, 2025. [Online]. Available:
  \url{https://doi.org/10.20935/acadquant7900}
\BIBentrySTDinterwordspacing

\bibitem{Shaydulin2021ClassicalSymmetriesAndTheQuantumApproximateOptimizationAlgorithm}
\BIBentryALTinterwordspacing
R.~Shaydulin, S.~Hadfield, T.~Hogg, and I.~Safro, ``{Classical symmetries and
  the Quantum Approximate Optimization Algorithm},'' \emph{Quantum Inf.
  Process.}, vol.~20, no.~11, 2021. [Online]. Available:
  \url{https://doi.org/10.1007/s11128-021-03298-4}
\BIBentrySTDinterwordspacing

\bibitem{Martonak2004QuantumAnnealingOfTheTravelingSalesmanProblem}
\BIBentryALTinterwordspacing
R.~Marto{\v{n}}{\'a}k, G.~E. Santoro, and E.~Tosatti, ``{Quantum annealing of
  the traveling-salesman problem},'' \emph{Phys. Rev.}, vol.~70, no.~5, 2004.
  [Online]. Available: \url{https://doi.org/10.1103/physreve.70.057701}
\BIBentrySTDinterwordspacing

\bibitem{Warren2019SolvingTheTravelingSalesmanProblemOnAQuantumAnnealer}
\BIBentryALTinterwordspacing
R.~H. Warren, ``{Solving the traveling salesman problem on a quantum
  annealer},'' \emph{SN Appl. Sci.}, vol.~2, no.~1, 2019. [Online]. Available:
  \url{https://doi.org/10.1007/s42452-019-1829-x}
\BIBentrySTDinterwordspacing

\bibitem{Chen2006OptimizedAnnealingOfTravelingSalesmanProblemFromTheNthNearestNeighborDistribution}
\BIBentryALTinterwordspacing
Y.~Chen and P.~Zhang, ``{Optimized annealing of traveling salesman problem from
  the nth-nearest-neighbor distribution},'' \emph{Phys. A: Stat. Mech. its
  Appl.}, vol. 371, no.~2, pp. 627--632, 2006. [Online]. Available:
  \url{https://doi.org/10.1016/j.physa.2006.04.052}
\BIBentrySTDinterwordspacing

\bibitem{Goswami2025SolvingTheTravellingSalesmanProblemUsingBlochSphereEncoding}
\BIBentryALTinterwordspacing
K.~Goswami, G.~Anekonda~Veereshi, P.~Schmelcher, and R.~Mukherjee, ``{Solving
  the travelling salesman problem using Bloch sphere encoding},'' \emph{Quantum
  Sci. Technol.}, vol.~11, no.~1, p. 015007, 2025. [Online]. Available:
  \url{https://doi.org/10.1088/2058-9565/ae1e9a}
\BIBentrySTDinterwordspacing

\bibitem{Grover1996AFastQuantumMechanicalAlgorithmForDatabaseSearch}
L.~K. Grover, ``{A fast quantum mechanical algorithm for database search},''
  1996, [arXiv preprint
  \href{https://arxiv.org/abs/quant-ph/9605043}{arXiv:quant-ph/9605043}].

\bibitem{Sato2025TwoStepQuantumSearchAlgorithmForSolvingTravelingSalesmanProblems}
\BIBentryALTinterwordspacing
R.~Sato, C.~Gordon, K.~Saito, H.~Kawashima, T.~Nikuni \emph{et~al.},
  ``{Two-Step Quantum Search Algorithm for Solving Traveling Salesman
  Problems},'' \emph{IEEE Trans. Quantum Eng.}, vol.~6, pp. 1--12, 2025.
  [Online]. Available: \url{https://doi.org/10.1109/tqe.2025.3548706}
\BIBentrySTDinterwordspacing

\bibitem{Bang2012AQuantumHeuristicAlgorithmForTheTravelingSalesmanProblem}
\BIBentryALTinterwordspacing
J.~Bang, J.~Ryu, C.~Lee, S.~Yoo, J.~Lim \emph{et~al.}, ``{A quantum heuristic
  algorithm for the traveling salesman problem},'' \emph{J. Korean Phys. Soc.},
  vol.~61, no.~12, pp. 1944--1949, 2012. [Online]. Available:
  \url{https://doi.org/10.3938/jkps.61.1944}
\BIBentrySTDinterwordspacing

\bibitem{Bartschi2020GroverMixersForQaoaShiftingComplexityFromMixerDesignToStatePreparation}
\BIBentryALTinterwordspacing
A.~Bartschi and S.~Eidenbenz, ``{Grover Mixers for QAOA: Shifting Complexity
  from Mixer Design to State Preparation},'' in \emph{2020 IEEE International
  Conference on Quantum Computing and Engineering (QCE)}, 2020, pp. 72--82.
  [Online]. Available: \url{https://doi.org/10.1109/qce49297.2020.00020}
\BIBentrySTDinterwordspacing

\bibitem{Binkowski2025QuantumFisherYatesShuffleUnifyingMethodsForGeneratingUniformSuperpositionsOfPermutations}
L.~Binkowski and M.~Schwiering, ``{Quantum Fisher-Yates shuffle: Unifying
  methods for generating uniform superpositions of permutations},'' 2025,
  [arXiv preprint \href{https://arxiv.org/abs/2504.17965}{arXiv:2504.17965}].

\bibitem{Hess2026GroverAdaptiveSearchWithProblemSpecificStatePreparation}
M.~Hess, L.~Palackal, A.~Awasthi, P.~J. Eder, M.~Schnaus \emph{et~al.},
  ``{Grover Adaptive Search with Problem-Specific State Preparation},'' 2026,
  [arXiv preprint \href{https://arxiv.org/abs/2602.08418}{arXiv:2602.08418}].

\bibitem{Binkowski2026QuantumAlgorithmsForCombinatorialOptimisation}
\BIBentryALTinterwordspacing
L.~Binkowski, ``{Quantum algorithms for combinatorial optimisation},'' 2026.
  [Online]. Available:
  \url{https://www.itp.uni-hannover.de/fileadmin/itp/qinfo/Team_Tobias_Osborne/Doctoral_Theses/Lennart_Binkowski_Doctoral_Thesis.pdf}
\BIBentrySTDinterwordspacing

\bibitem{Schwiering2024QuantumOptimizationAlgorithmsForTheTravelingSalesmanProblem}
\BIBentryALTinterwordspacing
M.~Schwiering, ``{Quantum Optimization Algorithms for the Traveling Salesman
  Problem},'' 2024. [Online]. Available:
  \url{https://www.itp.uni-hannover.de/fileadmin/itp/qinfo/Team_Tobias_Osborne/Masters_Theses/Marvin_Schwiering_Masters_Thesis.pdf}
\BIBentrySTDinterwordspacing

\bibitem{Friend1956SortingOnElectronicComputerSystems}
\BIBentryALTinterwordspacing
E.~H. Friend, ``{Sorting on Electronic Computer Systems},'' \emph{J. ACM},
  vol.~3, no.~3, pp. 134--168, 1956. [Online]. Available:
  \url{https://doi.org/10.1145/320831.320833}
\BIBentrySTDinterwordspacing

\bibitem{Mannila1985MeasuresOfPresortednessAndOptimalSortingAlgorithms}
\BIBentryALTinterwordspacing
H.~Mannila, ``{Measures of Presortedness and Optimal Sorting Algorithms},''
  \emph{IEEE Trans. Comput.}, vol. C-34, no.~4, pp. 318--325, 1985. [Online].
  Available: \url{https://doi.org/10.1109/tc.1985.5009382}
\BIBentrySTDinterwordspacing

\bibitem{Ziegler2026OrkanCacheFriendlySimulationOfQuantumOperationsOnHermitianOperators}
T.~Ziegler, ``{Orkan: Cache-friendly simulation of quantum operations on
  hermitian operators},'' 2026, [arXiv preprint
  \href{https://arxiv.org/abs/2604.15765}{arXiv:2604.15765}].

\end{thebibliography}

\end{document}